\documentclass[11pt]{article}

\usepackage[margin=1in]{geometry}
\usepackage[T1]{fontenc}

\usepackage{amsmath,amssymb,amsthm,mathtools}
\usepackage{booktabs}
\usepackage{enumitem}
\usepackage[hidelinks]{hyperref}
\hypersetup{
  pdftitle={Fair Division with Strictly Increasing Valuations: A Tight Threshold for Two-Agent EF1 and PO},
  pdfauthor={Nicholas Teh}
}
\usepackage{microtype}
\usepackage{authblk}
\usepackage{xcolor}
\usepackage[sort,numbers]{natbib}
\newtheorem{theorem}{Theorem}[section]
\newtheorem{lemma}[theorem]{Lemma}

\newtheorem{remark}[theorem]{Remark}
\theoremstyle{definition}
\newtheorem{definition}[theorem]{Definition}

\hypersetup{
     colorlinks   = true,
     linkcolor    = red, 
     urlcolor     = teal, 
	 citecolor    = blue 
}

\newcommand{\EFone}{\ensuremath{\text{EF1}}}
\newcommand{\PO}{\ensuremath{\text{PO}}}
\newcommand{\wPO}{\ensuremath{\text{weak\text{-}PO}}}
\newcommand{\calB}{\mathcal{B}}
\newcommand{\calU}{\mathcal{U}}
\newcommand{\calF}{\mathcal{F}}
\newcommand{\cov}{\mathrm{cov}}
\newcommand{\proofpart}[1]{\par\medskip\noindent\textbf{#1}\enspace}

\title{\bfseries Fair Division with Strictly Increasing Valuations: \\ A Tight Threshold for Two-Agent EF1 and PO}
\author{Nicholas Teh}
\affil[]{University of Oxford, UK}
\date{}

\begin{document}
\maketitle

\begin{abstract}
We study whether strictly positive marginal values restore the compatibility of envy-freeness up to one good (\EFone) and Pareto optimality (\PO) for indivisible goods.
For two agents, we identify the exact threshold in the number of goods. Every instance with at most seven goods and strictly increasing valuations admits an allocation that is both \EFone{} and \PO, without any submodularity assumption.
In contrast, we construct an eight-good instance with normalized, integer-valued, strictly increasing, submodular valuations in which every \EFone{} allocation is strictly Pareto dominated. Thus, eight goods are necessary and sufficient for a two-agent counterexample.
Finally, we strengthen the three-agent NP-hardness result of Chandramouleeswaran and Nimbhorkar (2026): deciding whether an \EFone{} and \PO{} allocation exists remains NP-hard for normalized, integer-valued, monotone submodular valuations even when zero marginals are confined to eight fixed agent-good pairs, all involving a single agent.
\end{abstract}

\section{Introduction}

Fair division of indivisible goods is a well-studied problem in economics, operations research, and computer science.
A central objective is to reconcile \emph{fairness} with \emph{economic efficiency}. 
Exact \emph{envy-freeness} is often unattainable: for example, a single desirable good cannot be allocated between two agents without creating envy. 
This motivates relaxations such as \emph{envy-freeness up to one good} (\EFone), which requires that any envy can be eliminated by removing some good from the envied bundle, and the stronger notion of \emph{envy-freeness up to any good} (EFX), which requires this condition to hold after the removal of any good from that bundle. Since a fair allocation may nevertheless admit a Pareto improvement, it is natural to ask whether these fairness guarantees can be achieved together with \emph{Pareto optimality} (\PO).

For additive valuations, an allocation that is both \EFone{} and \PO{} always exists~\cite{CaragiannisKMPPSW19}.
The same is true for matroid-rank valuations~\cite{BCIZ21}.
For general monotone submodular valuations, however, recent work gives counterexamples.
Mackenzie and Suzuki~\cite{MS26} construct a two-agent, eight-good instance in which every \EFone{} allocation is strictly Pareto dominated, while Chandramouleeswaran and Nimbhorkar~\cite{CN26} give a two-agent, six-good coverage instance with no allocation that is both \EFone{} and \PO. Both constructions rely on zero marginal values.
This motivates Open Problem~3.4 of Chandramouleeswaran and Nimbhorkar~\cite{CN26}: 
\begin{quote}
    \emph{Does every instance with strictly positive marginal values admit an allocation that is simultaneously \EFone{} and \PO?}
\end{quote}
They also show that deciding whether an \EFone{} and \PO{} allocation exists is NP-hard for three agents with monotone submodular valuations~\cite[Theorem~3]{CN26}.

Our first contribution answers the open problem negatively and determines the exact two-agent threshold in the number of goods. We modify the eight-good construction of Mackenzie and Suzuki~\cite{MS26} to obtain normalized, integer-valued, strictly increasing, submodular valuations for which every \EFone{} allocation is strictly Pareto dominated. We complement this counterexample with a matching existence result: every two-agent instance with at most seven goods and strictly increasing valuations admits an allocation that is both \EFone{} and \PO. This positive result requires no submodularity assumption. Together, these results show that eight goods are necessary and sufficient for a two-agent counterexample.

Our second contribution strengthens the three-agent hardness result of Chandramouleeswaran and Nimbhorkar~\cite{CN26}. We show that deciding whether an \EFone{} and \PO{} allocation exists remains NP-hard for normalized, integer-valued, monotone submodular valuations even when zero marginals are confined to eight fixed agent-good pairs, all involving a single agent.

Section~\ref{sec:eight} presents the eight-good counterexample. Section~\ref{sec:seven} proves existence for two agents with at most seven goods. Section~\ref{sec:hardness} establishes the restricted NP-hardness result.

\subsection{Related Work}
\EFone{} was formalized by Budish~\cite{Budish2011} and is implicit in the envy cycle elimination algorithm of Lipton et al.~\cite{LiptonMMS04}, which computes an \EFone{} allocation for arbitrary monotone valuations. This general existence guarantee does not impose efficiency. For additive valuations, Caragiannis et al.~\cite{CaragiannisKMPPSW19} showed that every maximum Nash welfare allocation is \EFone{} and \PO. Barman, Krishnamurthy, and Vaish~\cite{BarmanKV18} subsequently gave a pseudopolynomial-time algorithm for computing such an allocation and established the stronger existence of an allocation that is \EFone{} and fractionally Pareto optimal. For a fixed number of agents, Mahara~\cite{Mahara2026} obtained a polynomial-time algorithm for \EFone{} and fractional Pareto optimality. We refer to Amanatidis et al.~\cite{AmanatidisABFLMVW23} for a broader survey of discrete fair division.

Beyond additive valuations, exact compatibility is known for several structured classes. Benabbou et al.~\cite{BCIZ21} showed that matroid-rank valuations admit, in polynomial time, an \EFone{} allocation maximizing utilitarian welfare, and hence an \EFone{} and \PO{} allocation. For binary valuations, Brandl, Suksompong, and Teh~\cite{Brandl2026} characterize the allocation rule induced by maximum Nash welfare and leximin through \EFone{}, strategyproofness, neutrality, and related axioms.

Strictly positive marginal values have also been studied in connection with the stronger EFX requirement. Plaut and Roughgarden~\cite{plaut2020almost} proved that, when all agents have the same monotone valuation and all marginal values are nonzero, a leximin allocation is EFX and \PO. They also showed that, for two agents with distinct general valuations, EFX and \PO{} can be incompatible even when all marginals are nonzero. For identical monotone submodular valuations, Chandramouleeswaran and Nimbhorkar~\cite{CN26} proved that every leximin allocation is \EFone{} and \PO, without requiring strict increase. These identical valuation results do not resolve the heterogeneous setting considered here. In the setting with only a few more goods than agents, Lim, Neoh, and Teh~\cite{Lim2026EFXcost} study the welfare and computational consequences of imposing EFX, including the complexity of achieving EFX together with Pareto optimality.

For broader valuation classes, efficiency becomes harder to reconcile with fairness. Caragiannis et al.~\cite{CaragiannisKMPPSW19} showed that maximum Nash welfare need not be \EFone{} for monotone submodular valuations; their positive guarantee is the weaker \emph{marginal \EFone{}} property. They also showed that \EFone{} and \PO{} may be incompatible for monotone subadditive valuations. Barman and Suzuki~\cite{BarmanSuzuki2026} recover compatibility with a multiplicative $1/2$-approximation to Pareto optimality for subadditive valuations.

Related questions have also been studied for chores and mixed manna. For two agents with additive valuations over goods and chores, Aziz et al.~\cite{aziz2022goodschores} give a polynomial-time algorithm for computing an \EFone{} and \PO{} allocation. For general monotone chore costs, Bhaskar, Sricharan, and Vaish~\cite{BhaskarEtAl2021} give a polynomial-time algorithm for \EFone{}, without an efficiency guarantee. Mahara~\cite{Mahara2026} proved that additive chore instances with any number of agents admit an \EFone{} and fractionally Pareto-optimal allocation, while Teh~\cite{Teh2026} showed that an \EFone{} and \PO{} allocation can be computed efficiently when the number of chores exceeds the number of agents by a constant.
Beyond additivity, incompatibility reappears. Hosseini, Narang, and Wąs~\cite{HosseiniEtAl2025} exhibit a two-agent instance with identical monotone submodular coverage costs, induced by a rooted tree, that has no \EFone{} and \PO{} allocation, and show that deciding existence is NP-hard even for unweighted trees. The six-good coverage construction of Chandramouleeswaran and Nimbhorkar~\cite{CN26} also gives a counterexample for chores.

\section{Preliminaries}

Let $N$ be a finite set of \emph{agents} and let $M$ be a finite set of \emph{indivisible goods}.
Each agent $i\in N$ has a \emph{valuation} function $v_i:2^M\to\mathbb{R}_{\ge 0}$, normalized so that $v_i(\varnothing)=0$.
An \emph{allocation} is a tuple $X=(X_i)_{i\in N}$ whose bundles are pairwise disjoint and whose union is $M$.
Thus, all allocations in this paper are \emph{complete} (i.e., all goods are allocated).

For $S\subseteq M$ and $g\in M\setminus S$, the \emph{marginal value} of $g$ given $S$ is $v_i(g\mid S)=v_i(S\cup\{g\})-v_i(S).$
A valuation function $v$ is \emph{monotone} if $v(S)\le v(T)$ whenever $S\subseteq T$.
It is \emph{strictly increasing} if $v(g\mid S)>0$ for every $S\subseteq M\setminus\{g\}$.
Since $M$ is finite, this is equivalent to $v(S)<v(T)$ whenever $S\subsetneq T$.
Accordingly, ``strictly increasing'' and ``having strictly positive marginals'' refer to the same condition in this paper.\footnote{This is exactly the condition in Open Problem~3.4 of Chandramouleeswaran and Nimbhorkar~\cite{CN26}.}
A valuation is \emph{submodular} if it has diminishing marginal returns, i.e., $v(g\mid S)\ge v(g\mid T)$ whenever $S\subseteq T\subseteq M\setminus\{g\}$.\footnote{For the computational result in Section~\ref{sec:hardness}, each valuation function is represented by a polynomial-size description from which $v_i(S)$ can be evaluated in polynomial time.}

Following Budish~\cite{Budish2011} and Lipton et al.~\cite{LiptonMMS04}, we use \EFone{} as our fairness criterion.

\begin{definition}[\EFone]
An allocation $X=(X_i)_{i\in N}$ is \emph{envy-free up to one good} (\EFone) if, for every ordered pair of distinct agents $i,j\in N$, either $X_j=\varnothing$, or there exists $g\in X_j$ such that $v_i(X_i)\ge v_i(X_j\setminus\{g\})$.
\end{definition}

We measure efficiency using the standard notions of Pareto dominance and Pareto optimality; see, e.g., Caragiannis et al.~\cite{CaragiannisKMPPSW19} and Barman, Krishnamurthy, and Vaish~\cite{BarmanKV18}.

\begin{definition}[Pareto notions]
An allocation $Y$ \emph{Pareto dominates} $X$ if $v_i(Y_i)\ge v_i(X_i)$ for every $i\in N$, with at least one strict inequality.
It \emph{strictly Pareto dominates} $X$ if $v_i(Y_i)>v_i(X_i)$ for every $i\in N$.
An allocation is \emph{Pareto optimal} (\PO) if it is not Pareto dominated.
It is \emph{weakly Pareto optimal} (\wPO) if it is not strictly Pareto dominated.
Thus \PO{} implies \wPO.
\end{definition}

\section{An Eight-Good Counterexample with Strictly Positive Marginals}\label{sec:eight}

We first establish the negative side of our two-agent threshold result by constructing an eight-good instance with strictly positive marginal values in which \EFone{} and Pareto efficiency are incompatible. The construction is obtained by perturbing the example of Mackenzie and Suzuki~\cite{MS26}: we fix $\varepsilon=1/6$, scale all values by $12$, and add $|S|$ to the value of every bundle $S$. The additive term makes every marginal value strictly positive, while the scaling ensures that the strict \EFone{} violations and strict Pareto improvements in the original instance survive the perturbation. We verify all required properties directly from the resulting integer-valued tables, so the proof is self-contained.

Let $M=A\mathbin{\dot\cup}B$, where $A=\{a_1,a_2,a_3\}$ and $B=\{b_1,b_2,b_3,b_4,b_5\}$.
For a bundle $S\subseteq M$, let $x(S)=|S\cap A|$ and $y(S)=|S\cap B|$.
Each agent's value for a bundle depends only on these two type counts.

\begin{equation}
\begin{array}{c|rrrrrr}
 u_1(x,y) & y=0&y=1&y=2&y=3&y=4&y=5\\ \hline
 x=0&0&13&26&39&46&53\\
 x=1&25&38&47&50&53&54\\
 x=2&28&41&50&53&54&55\\
 x=3&29&42&51&54&55&56
\end{array}
\label{eq:u1-table}
\end{equation}

\begin{equation}
\begin{array}{c|rrrrrr}
 u_2(x,y) & y=0&y=1&y=2&y=3&y=4&y=5\\ \hline
 x=0&0&13&26&39&50&51\\
 x=1&17&28&39&46&53&54\\
 x=2&32&39&46&53&54&55\\
 x=3&47&50&53&54&55&56
\end{array}
\label{eq:u2-table}
\end{equation}

For each $i\in\{1,2\}$, define $v_i(S)=u_i(x(S),y(S))$ for every $S \subseteq M$. Our result is as follows.

\begin{theorem}\label{thm:upper}
    The valuations $v_1$ and $v_2$ are normalized, integer-valued, strictly increasing, and submodular.
    Every \EFone{} allocation is strictly Pareto dominated.
    Consequently, no allocation is simultaneously \EFone{} and \wPO, and in particular no allocation is simultaneously \EFone{} and \PO.
\end{theorem}

We use the following elementary criterion for count-based valuations.

\begin{lemma}\label{lem:count-criterion}
Let $M=A\mathbin{\dot\cup}B$, where $|A|=p$ and $|B|=q$, and let $v(S)=f(|S\cap A|,|S\cap B|)$.
Define the $A$- and $B$-marginal arrays by
\begin{align*}
\Delta_A f(x,y)&=f(x+1,y)-f(x,y)
&& (0\le x<p,\ 0\le y\le q),\\
\Delta_B f(x,y)&=f(x,y+1)-f(x,y)
&& (0\le x\le p,\ 0\le y<q).
\end{align*}
Suppose every entry of both arrays is positive and each array is coordinatewise nonincreasing: whenever both entries are defined and $x\le x'$, $y\le y'$,
\begin{equation*}
\Delta_A f(x,y)\ge \Delta_A f(x',y')
\quad\text{and}\quad
\Delta_B f(x,y)\ge \Delta_B f(x',y').
\end{equation*}
Then $v$ is strictly increasing and submodular.
\end{lemma}

\begin{proof}
At a bundle with type counts $(x,y)$, the marginal value of any unallocated $A$-good is $\Delta_A f(x,y)$, and the marginal value of any unallocated $B$-good is $\Delta_B f(x,y)$.
Positivity therefore gives strict increase.

For submodularity, let $S\subseteq T$ and $g\notin T$.
If $g\in A$, let $(x,y)$ and $(x',y')$ be the type counts of $S$ and $T$, respectively.
Then $x\le x'$ and $y\le y'$, so the assumed coordinatewise monotonicity gives
\begin{equation*}
v(g\mid S)=\Delta_A f(x,y)
\ge
\Delta_A f(x',y')=v(g\mid T).
\end{equation*}
The argument for $g\in B$ is identical.
\end{proof}

\begin{proof}[Proof of Theorem~\ref{thm:upper}]
\proofpart{Valuation properties.}
The tables give $v_i(\varnothing)=u_i(0,0)=0$, and all entries are integers.
Appendix~\ref{app:marginals} displays every value of $\Delta_A u_i$ and $\Delta_B u_i$.
Each entry is at least $1$, and in each array every row and every column is nonincreasing.
Lemma~\ref{lem:count-criterion} therefore proves that both valuations are strictly increasing and submodular.

\proofpart{\EFone{} splits.}
Represent an allocation by a \emph{count split} $(x,y)$: agent~1 receives $x$ goods from $A$ and $y$ goods from $B$, while agent~2 receives the remaining $(3-x,5-y)$ goods.
Since the agents are indifferent among goods of the same type, all allocations with the same count split have the same utilities and the same \EFone{} status.

Consider agent~1's comparison with agent~2.
Deleting an $A$-good from agent~2's bundle, when possible, leaves type counts $(2-x,5-y)$; deleting a $B$-good, when possible, leaves $(3-x,4-y)$.
Let $d_1(x,y)$ be the smaller $u_1$-value among the applicable deletion outcomes.
When agent~2's bundle is nonempty, agent~1 satisfies \EFone{} exactly when
\begin{equation}
u_1(x,y)\ge d_1(x,y).
\label{eq:agent1-ef1}
\end{equation}
If agent~2's bundle is empty, this comparison is automatic.

Similarly, deleting an $A$-good from agent~1's bundle leaves counts $(x-1,y)$, and deleting a $B$-good leaves $(x,y-1)$.
Let $d_2(x,y)$ be the smaller $u_2$-value among the applicable outcomes.
When agent~1's bundle is nonempty, agent~2 satisfies \EFone{} exactly when
\begin{equation}
u_2(3-x,5-y)\ge d_2(x,y).
\label{eq:agent2-ef1}
\end{equation}
If agent~1's bundle is empty, this comparison is automatic.

For each nonautomatic comparison, define
\begin{equation*}
\sigma_1(x,y)=u_1(x,y)-d_1(x,y),
\quad
\sigma_2(x,y)=u_2(3-x,5-y)-d_2(x,y).
\end{equation*}
Thus the relevant agent satisfies \EFone{} if and only if the corresponding $\sigma_i(x,y)$ is nonnegative.
Substituting the entries of \eqref{eq:u1-table}--\eqref{eq:u2-table} yields the following classification.
An entry $E$ marks a split that is \EFone{} for both agents; an entry $1$ or $2$ identifies the unique agent who violates \EFone{} at that split.

\begin{equation}
\begin{array}{c|rrrrrr}
 &y=0&y=1&y=2&y=3&y=4&y=5\\ \hline
x=0&1&1&1&1&E&2\\
x=1&1&1&1&E&2&2\\
x=2&1&1&E&2&2&2\\
x=3&1&E&2&2&2&2
\end{array}
\label{eq:ef1-classification}
\end{equation}

Appendix~\ref{app:ef1-differences} gives the underlying values of $\sigma_1$ and $\sigma_2$.
The only \EFone{} count splits are therefore
\begin{equation}
(0,4),\quad (1,3),\quad (2,2),\quad (3,1).
\label{eq:ef1-splits}
\end{equation}
At each of these splits, both agents receive four goods.

\proofpart{Strict Pareto domination.}
At split $(x,y)$, the utility pair is $(u_1(x,y),u_2(3-x,5-y))$.
For every \EFone{} split, the following table exhibits another count split that strictly improves both coordinates.

\begin{equation}
\begin{array}{c@{\quad}c@{\quad}c@{\quad}c}
\toprule
\text{\EFone{} split}&\text{utility pair}&\text{dominating split}&\text{new utility pair}\\
\midrule
(0,4)&(46,50)&(1,2)&(47,53)\\
(1,3)&(50,46)&(0,5)&(53,47)\\
(2,2)&(50,46)&(0,5)&(53,47)\\
(3,1)&(42,50)&(1,2)&(47,53)\\
\bottomrule
\end{array}
\label{eq:dominating-splits}
\end{equation}
Hence every \EFone{} allocation is strictly Pareto dominated and therefore fails \wPO.
\end{proof}

\begin{remark}[How the example is obtained]\label{rem:construction}
Let $g_i$ denote agent $i$'s valuation in~\cite[Theorem~3.1]{MS26} with $\varepsilon=1/6$.
The integer tables satisfy $v_i(S)=12g_i(S)+|S|$.
The term $|S|$ adds exactly $1$ to the marginal value of every good.

In the original construction, \EFone{} holds exactly at the balanced $4$--$4$ splits.
At every other split, the agent who violates \EFone{} holds at most three goods, while the envied bundle still contains at least four goods after one deletion.
Adding $|S|$ therefore decreases the violating agent's own-minus-envied comparison by at least one.
At a balanced split, the comparison is between a four-good own bundle and a three-good bundle after deletion, so the same term increases the comparison by one and preserves \EFone{}.

Each displayed Pareto improvement changes the bundle sizes from $4$--$4$ to $3$--$5$, with the identity of the agent receiving the smaller bundle depending on the split.
After multiplication by $12$, each increase is at least $2$.
The agent whose bundle shrinks loses only one unit from the $|S|$ term, while the other agent gains one unit.
Both agents therefore still improve strictly.
The direct verification above and in Appendices~\ref{app:marginals}--\ref{app:ef1-differences} makes the transformed instance independently checkable.
\end{remark}

Every marginal in Theorem~\ref{thm:upper} is positive, yet the instance has no \EFone{} and \PO{} allocation, thus provides a negative answer to an open problem of Chandramouleeswaran and Nimbhorkar~\cite{CN26}, that strictly positive marginal values do not guarantee an allocation that is simultaneously \EFone{} and \PO, even for two agents with submodular valuations.

\section{Existence for Two Agents with at Most Seven Goods}\label{sec:seven}

Having shown that eight goods suffice for nonexistence, we now prove that the construction of Section~\ref{sec:eight} is minimal. Specifically, every two-agent instance with at most seven goods and strictly increasing valuations admits an allocation that is both \EFone{} and \PO. Notably, this positive result does not require submodularity.

The proof associates with each agent a threshold separating bundles that may violate \EFone{} from those that are automatically safe. We then show that, unless the ground set contains at least eight goods, the two agents can be assigned disjoint bundles lying above their respective thresholds. Such an allocation is \EFone{}, and a Pareto improvement argument then yields an allocation that is simultaneously \EFone{} and \PO.

The following terminology is specific to the two-agent setting. A bundle is called \EFone-violating only with respect to the allocation in which the other agent receives its complement.

\begin{definition}[\EFone{}-violating bundles and the violation threshold]\label{def:violating-threshold}
Fix a valuation $v$ on $M$.
For $S\subseteq M$, write $T=M\setminus S$.
We call $S$ \emph{\EFone{}-violating} for $v$ if $T\ne\varnothing$ and
\begin{equation}
v(S)<v(T\setminus\{g\})
\quad
\text{for every }g\in T.
\label{eq:violating-bundle}
\end{equation}
Equivalently, an agent receiving $S$ in the allocation $(S,T)$ violates \EFone{}.

Let $\calB(v)$ be the family of \EFone{}-violating bundles.
When $\calB(v)\ne\varnothing$, define $\tau(v)=\max_{S\in\calB(v)}v(S)$ and $\calU(v)=\{R\subseteq M:v(R)>\tau(v)\}$.
No bundle in $\calU(v)$ is \EFone{}-violating.
\end{definition}

\begin{lemma}\label{lem:threshold-families}
Let $v$ be strictly increasing and suppose $\calB(v)\ne\varnothing$.
Choose $S\in\calB(v)$ with $v(S)=\tau(v)$, and let $T=M\setminus S$.
Then $|T|\ge2$, the family $\calU(v)$ is closed under supersets, and
\begin{equation}
\calF(S,T)
=
\{S\cup\{g\}:g\in T\}
\cup
\{T\setminus\{g\}:g\in T\}
\subseteq \calU(v).
\label{eq:threshold-family}
\end{equation}
\end{lemma}

\begin{proof}
If $|T|=1$, say $T=\{g\}$, then \eqref{eq:violating-bundle} gives $v(S)<v(\varnothing)$,
contradicting monotonicity because $\varnothing\subseteq S$.
Thus $|T|\ge2$.

If $R\in\calU(v)$ and $R\subseteq R'$, monotonicity gives $v(R')\ge v(R)>\tau(v)$, so $R'\in\calU(v)$.
Hence $\calU(v)$ is closed under supersets.

Finally, for each $g\in T$, strict increase gives $v(S\cup\{g\})>v(S)=\tau(v)$, while \eqref{eq:violating-bundle} gives $v(T\setminus\{g\})>v(S)=\tau(v)$.
Both kinds of bundles in \eqref{eq:threshold-family} therefore belong to $\calU(v)$.
\end{proof}

Two set families $\mathcal G_1,\mathcal G_2\subseteq 2^M$ are \emph{cross-intersecting} if every set in $\mathcal G_1$ intersects every set in $\mathcal G_2$.
Cross-intersection is exactly the obstruction that prevents one from choosing disjoint bundles, one from each family.
The next lemma shows that the structured families in Lemma~\ref{lem:threshold-families} can have this obstruction only on a ground set of at least eight goods.

\begin{lemma}\label{lem:cross-count}
Let $M=S_1\mathbin{\dot\cup}T_1=S_2\mathbin{\dot\cup}T_2$ for any $|T_1|,|T_2|\ge2$ and define $\calF_i=\calF(S_i,T_i)$ as in \eqref{eq:threshold-family}.
Let $p,q,r,s$ be the sizes of the four cells formed by the two bipartitions:
\begin{equation}
\begin{array}{c|cc}
 & S_2 & T_2\\ \hline
S_1 & p=|S_1\cap S_2| & q=|S_1\cap T_2|\\
T_1 & r=|T_1\cap S_2| & s=|T_1\cap T_2|
\end{array}.
\label{eq:four-cells}
\end{equation}
Then $\calF_1$ and $\calF_2$ are cross-intersecting if and only if
\begin{equation}
p\ge1,\quad q\ge2,\quad r\ge2,\quad s\ge3.
\label{eq:cell-bounds}
\end{equation}
In particular, cross-intersection implies $|M|\ge8$.
\end{lemma}

\begin{proof}
Assume first that $\calF_1$ and $\calF_2$ are cross-intersecting.

We begin with $s$.
If $s\le2$, choose $g\in T_1$ and $h\in T_2$ so that $T_1\cap T_2\subseteq\{g,h\}$.
If the intersection is empty, choose $g$ and $h$ arbitrarily.
If it has one element, choose that element for both.
If it has two elements, choose one as $g$ and the other as $h$.
In every case, $(T_1\setminus\{g\})\cap(T_2\setminus\{h\})=\varnothing$, contradicting cross-intersection.
Hence $s\ge3$.

Next suppose $p=0$.
Choose distinct $g,h\in T_1\cap T_2$, which is possible because $s\ge3$.
The goods $g$ and $h$ lie outside both $S_1$ and $S_2$, and $g\ne h$.
Therefore, $(S_1\cup\{g\})\cap(S_2\cup\{h\})=\varnothing$,
again a contradiction.
Thus $p\ge1$.

We now prove $q\ge2$.
If $q=0$, choose $x\in T_1\cap T_2$.
Then $(S_1\cup\{x\})\cap(T_2\setminus\{x\})=\varnothing$, contradicting cross-intersection.

Suppose instead that $q=1$, and write $S_1\cap T_2=\{x\}$.
The set $T_2\setminus\{x\}$ is disjoint from $S_1$.
For every $g\in T_1$, it must nevertheless intersect $S_1\cup\{g\}$.
The only possible intersection point is $g$, so $T_1\subseteq T_2\setminus\{x\}$.
It follows that $T_1\cap S_2=\varnothing$, hence $S_2\subseteq S_1$.
Moreover, $S_1\setminus S_2=S_1\cap T_2=\{x\}$, so $S_1=S_2\cup\{x\}\in\calF_2$.
But $S_1$ is disjoint from every set $T_1\setminus\{g\}\in\calF_1$, a contradiction.
Therefore $q\ge2$.
By symmetry, $r\ge2$.
This proves the necessity of \eqref{eq:cell-bounds}.

Conversely, assume \eqref{eq:cell-bounds}.
Two sets of the forms $S_1\cup\{g\}$ and $S_2\cup\{h\}$ intersect in $S_1\cap S_2$, which is nonempty because $p\ge1$.
A set $S_1\cup\{g\}$ intersects every $T_2\setminus\{h\}$ because deleting one good cannot remove all $q\ge2$ goods in $S_1\cap T_2$.
The symmetric statement follows from $r\ge2$.
Finally, $(T_1\setminus\{g\})\cap(T_2\setminus\{h\})$ is obtained from $T_1\cap T_2$ by deleting at most two goods, and is therefore nonempty because $s\ge3$.
Thus $\calF_1$ and $\calF_2$ are cross-intersecting.

The four cells in \eqref{eq:four-cells} partition $M$.
Hence cross-intersection implies
\begin{equation*}
|M|=p+q+r+s\ge1+2+2+3=8. \qedhere
\end{equation*}
\end{proof}

\begin{theorem}\label{thm:lower}
Every two-agent instance with at most seven goods and strictly increasing valuations admits an allocation that is simultaneously \EFone{} and \PO.
No submodularity assumption is needed.
\end{theorem}

\begin{proof}
If $|M|\le1$, the claim is immediate: every allocation is \EFone{}, and assigning the only good, if any, to either agent is \PO.
Assume henceforth that $2\le |M|\le7$.

For either agent $i$, the empty bundle is \EFone{}-violating.
Indeed, for every $g\in M$, the bundle $M\setminus\{g\}$ is nonempty, so strict increase and normalization give $v_i(M\setminus\{g\})>v_i(\varnothing)=0$.
Thus $\calB(v_i)\ne\varnothing$ for $i=1,2$.

For each agent $i$, choose $S_i\in\calB(v_i)$ with $v_i(S_i)=\tau(v_i)$ and $T_i=M\setminus S_i$.
Let $\calU_i=\calU(v_i)$ and $\calF_i=\calF(S_i,T_i)$.
By Lemma~\ref{lem:threshold-families}, $\calF_i\subseteq\calU_i$ and $|T_i|\ge2$.

Suppose, for a contradiction, that no allocation gives both agents a bundle in their respective families $\calU_i$.
Then $\calU_1$ and $\calU_2$ must be cross-intersecting.
To see this, suppose $R_1\in\calU_1$ and $R_2\in\calU_2$ were disjoint.
Assign $R_2$ to agent~2 and assign all remaining goods to agent~1.
Agent~1 then receives $M\setminus R_2\supseteq R_1$ which belongs to $\calU_1$ because $\calU_1$ is closed under supersets.
This would be an allocation with both bundles above their thresholds, contrary to the assumption.

Hence the subfamilies $\calF_1$ and $\calF_2$ are also cross-intersecting.
Lemma~\ref{lem:cross-count} would then imply $|M|\ge8$, contradicting $|M|\le7$.
Therefore there exists an allocation $X=(X_1,X_2)$ such that
\begin{equation}
v_i(X_i)>\tau(v_i)
\quad (i=1,2).
\label{eq:above-threshold-allocation}
\end{equation}
Neither $X_i$ is \EFone{}-violating, so $X$ is \EFone{}.

It remains to impose Pareto optimality without crossing either threshold.
Among all allocations $Z$ satisfying $v_i(Z_i)\ge v_i(X_i)$ for all $i = 1,2$, choose one, say $Y$, that is Pareto maximal within this finite collection.
If some allocation $W$ Pareto dominated $Y$, then $W$ would satisfy the same lower bounds and belong to the collection, contradicting the choice of $Y$.
Thus $Y$ is globally \PO.
Moreover, $v_i(Y_i)\ge v_i(X_i)>\tau(v_i)$, so neither $Y_i$ is \EFone{}-violating.
Hence $Y$ is also \EFone{}.

It remains to impose Pareto optimality without crossing either threshold.
Among all allocations $Z$ satisfying $v_i(Z_i) \ge v_i(X_i)$ for all $i = 1,2$,
choose an allocation $Y$ maximizing $v_1(Y_1)+v_2(Y_2)$
Such a maximizer exists because there are finitely many allocations.
If an allocation $W$ Pareto dominated $Y$, then $W$ would satisfy
the same lower bounds and $v_1(W_1)+v_2(W_2)
> v_1(Y_1)+v_2(Y_2)$,
contradicting the choice of $Y$. Thus $Y$ is globally PO.
Moreover, $v_i(Y_i) \ge v_i(X_i) > \tau(v_i)$ for $i = 1,2$, so neither $Y_i$ is \EFone{}-violating. Hence $Y$ is also \EFone{}.
\end{proof}

Thus, for two agents with strictly increasing valuations, eight is the minimum number of goods in an instance with no allocation that is both \EFone{} and \PO.
Moreover, an eight-good counterexample exists with normalized, integer-valued, strictly increasing, submodular valuations, and every \EFone{} allocation in that instance is strictly Pareto dominated.

\section{NP-Hardness with Only Eight Zero-Marginal Agent--Good Pairs}\label{sec:hardness}

We next turn from the exact two-agent threshold to the computational complexity of the three-agent problem. Chandramouleeswaran and Nimbhorkar~\cite{CN26} show that deciding whether an \EFone{} and \PO{} allocation exists is NP-hard for three agents with monotone submodular valuations. In their reduction, however, one agent assigns zero marginal value to every vertex good, so the number of zero-marginal agent--good pairs grows with the input.

We strengthen this result by confining all zero marginals to a fixed, constant-size core. Our reduction uses the eight goods from Section~\ref{sec:eight} as the core and introduces one vertex good for each vertex of the input graph. Every vertex good has strictly positive marginal value for every agent; the only zero marginals occur for agent~3 on the eight core goods. The core forces the unique split compatible with the relevant \EFone{} and \PO{} constraints, while the vertex goods encode a balanced vertex cover.

\begin{theorem}[Hardness with only eight zero-marginal pairs]\label{thm:hardness}
It is NP-hard to decide whether a three-agent instance with normalized, integer-valued, monotone submodular valuations admits an allocation that is both \EFone{} and \PO.
This remains true under the following restriction.
There is a fixed set $C$ of eight goods such that, for every bundle $R\subseteq M$ and every good $g\in M\setminus R$:
\begin{enumerate}[label=\textup{(\roman*)},leftmargin=2.5em,itemsep=0.15em,topsep=0.3em]
\item if $g\in M\setminus C$, then $v_i(g\mid R)>0$ for every agent $i$;
\item if $g\in C$, then $v_i(g\mid R)>0$ for $i\in\{1,2\}$; and
\item if $g\in C$, then $v_3(g\mid R)=0$.
\end{enumerate}
Thus the only agent--good pairs with identically zero marginal value are the eight pairs $(3,g)$ with $g\in C$.
\end{theorem}

\begin{proof}
We reduce from \textsc{Balanced Vertex Cover}~\cite[Lemma~1]{CS06}.
An instance is a graph $G=(V,E)$ with an even number $n=|V|$ of vertices, and the question is whether $G$ has a vertex cover of size at most $n/2$.
Conitzer and Sandholm state the problem with a cover of size exactly $n/2$.
The two formulations are equivalent because any cover of size at most $n/2$ can be enlarged to size exactly $n/2$.

If $E=\varnothing$, the source instance is trivially a yes-instance. We map it to the fixed output obtained by applying the construction to a single-edge graph, which is also a yes-instance. Henceforth, assume $m=|E|\ge1$.

\proofpart{Construction.}
Let $C=A\mathbin{\dot\cup}B$ be a disjoint copy of the eight goods from Section~\ref{sec:eight}, with $|A|=3$ and $|B|=5$.
For each vertex $v\in V$, add a good $p_v$.
Let $P=\{p_v:v\in V\}$ and $M=C\mathbin{\dot\cup}P$.
We call the goods in $C$ \emph{core goods} and those in $P$ \emph{vertex goods}.

For $R\subseteq M$, define $x(R)=|R\cap A|$ and $y(R)=|R\cap B|$.

For $Q\subseteq P$, let $\cov(Q) = |\{\{u,v\}\in E:p_u\in Q\text{ or }p_v\in Q\}|$.
Thus $\cov(Q)$ is the number of graph edges covered by the vertices represented in $Q$.

Set $L=n+2$.
Since $m\ge1$, we will use $Lm>n$ and $3L>n+1$.
For every $R\subseteq M$, define
\begin{align*}
w_1(R)
&=
L(m\,u_1(x(R),y(R))+3\cov(R\cap P))
+|R\cap P|,\\
w_2(R)
&=
Lm\,u_2(x(R),y(R))
+|R\cap P|,\\
w_3(R)
&=
L|R\cap P|.
\end{align*}
These are polynomial-size descriptions, and all three values can be computed in polynomial time from $G$ and the fixed tables \eqref{eq:u1-table}--\eqref{eq:u2-table}.

\proofpart{Valuation properties.}
On the core goods, the functions $D\mapsto u_i(x(D),y(D))$ are normalized, integer-valued, strictly increasing, and submodular by Theorem~\ref{thm:upper}.
Extending them to $M$ by ignoring vertex goods preserves normalization, integrality, monotonicity, and submodularity.
The coverage function $\cov$ is normalized, monotone, and submodular: as the current set grows, a newly added vertex good can cover only fewer previously uncovered edges.
The remaining terms are additive.
Hence $w_1,w_2,w_3$ are normalized, integer-valued, monotone, and submodular.

For a vertex good $p_v$, its marginal value is
\begin{equation*}
3L\cdot
(\text{number of newly covered edges})+1
\end{equation*}
for agent~1, $1$ for agent~2, and $L$ for agent~3.
All three are positive.
For a core good $g\in C$, every core marginal in \eqref{eq:u1-table}--\eqref{eq:u2-table} is at least $1$, so whenever $g\notin R$,
\begin{equation*}
w_1(g\mid R)\ge Lm,
\quad
w_2(g\mid R)\ge Lm,
\quad
w_3(g\mid R)=0.
\end{equation*}
This proves the promised marginal restriction.

\proofpart{Two facts about the core.}
We will use the following two consequences of the fixed tables.
\begin{enumerate}[label=\textup{(C\arabic*)},leftmargin=3em,itemsep=0.8em]
\item\label{core:frontier}
For a core split $(x,y)$, agent~1 receives counts $(x,y)$ and agent~2 receives $(3-x,5-y)$.
Among the $24$ possible count splits, the Pareto-optimal utility pairs are exactly the eight rows below.
The last two columns give the core \EFone{} differences $\sigma_1,\sigma_2$ defined in the proof of Theorem~\ref{thm:upper}; $\star$ denotes an automatic comparison with an empty envied bundle.
\begin{equation*}
\begin{array}{c|c|r|r}
\text{split }(x,y)
&(u_1(x,y),u_2(3-x,5-y))
&\sigma_1(x,y)&\sigma_2(x,y)\\ \hline
(0,0)&(0,56)&-55&\star\\
(1,0)&(25,55)&-29&55\\
(1,1)&(38,54)&-15&41\\
(1,2)&(47,53)&-3&27\\
(0,5)&(53,47)&25&-3\\
(1,5)&(54,32)&29&-19\\
(2,5)&(55,17)&55&-37\\
(3,5)&(56,0)&\star&-55
\end{array}
\end{equation*}
This is a direct comparison of the constant-size list of utility pairs.

\item\label{core:rigidity}
If $D_1,D_2\subseteq C$ are disjoint and
\begin{equation*}
u_1(x(D_1),y(D_1))\ge47,
\quad
u_2(x(D_2),y(D_2))\ge53,
\end{equation*}
then $D_1$ has counts $(1,2)$, $D_2$ has counts $(2,3)$, and $D_1\mathbin{\dot\cup}D_2=C$.

Indeed, the first inequality implies either $x(D_1)=0,y(D_1)=5$, or $x(D_1)\ge1,y(D_1)\ge2$.
The first possibility leaves no $B$-goods for a disjoint bundle of $u_2$-value at least $53$.
In the second possibility, disjointness gives $x(D_2)\le2$ and $y(D_2)\le3$.
The $u_2$-table then forces $x(D_2)=2,y(D_2)=3$.
Disjointness consequently forces $x(D_1)=1,y(D_1)=2$, and the two bundles exhaust $C$.
\end{enumerate}

At the key split $(1,2)$, agent~1's core \EFone{} difference is $-3$.
After the core values are multiplied by $m$, this is a deficit of $3m$, while covering all $m$ edges contributes exactly $3m$ through the coverage term.
This is why the coefficient $3$ appears in $w_1$.

\proofpart{Forward direction.}
Assume $G$ has a vertex cover $K\subseteq V$ with $k=|K|\le n/2$.
Let $P_K=\{p_v:v\in K\}$.
Choose a core bundle $C_1\subseteq C$ containing one $A$-good and two $B$-goods, and let $C_2=C\setminus C_1$.
Consider the allocation $X = (X_1,X_2,X_3)$ defined by
\begin{equation*}
X_1=C_1\cup P_K,
\quad
X_2=C_2,
\quad
X_3=P\setminus P_K.
\end{equation*}
Since $K$ covers all $m$ edges, the utilities are
\begin{equation*}
    w_1(X_1)=50Lm+k,
    \quad
    w_2(X_2)=53Lm,
    \quad
    w_3(X_3)=L(n-k).
\end{equation*}
We check all six ordered \EFone{} comparisons.
\begin{itemize}
    \item For agent~1's comparison with agent~2, deleting any core good from $C_2$ leaves a core bundle of $u_1$-value $50$, so $w_1(X_1)\ge50Lm=w_1(X_2\setminus\{g\})$.
    \item For agent~2's comparison with agent~1, delete the unique $A$-good $g\in C_1$.
Then $w_2(X_1\setminus\{g\})
= 26Lm+k
< 53Lm
= w_2(X_2)$, 
where the strict inequality follows from $k\le n<Lm$.
    \item Since $m\ge1$, every vertex cover is nonempty.
Agent~3 can therefore delete a vertex good from $P_K$, after which she values agent~1's bundle at $L(k-1)$.
Her own value satisfies $L(n-k)\ge L(k-1)$ because $k\le n/2$.
The remaining comparisons hold without any deletion: $w_1(X_3)\le3Lm+n<50Lm+k$, $w_2(X_3)\le n<53Lm$, and $w_3(X_2)=0\le L(n-k)$.
\end{itemize}
Thus $X$ is \EFone{}.

We next show that $X$ is \PO.
Suppose an allocation $Y$ Pareto dominates $X$.
To give agent~2 utility at least $53Lm$, the core part of $Y_2$ must have $u_2$-value at least $53$: with core value at most $52$, even all $n$ vertex goods yield at most $52Lm+n<53Lm$.
Similarly, the core part of $Y_1$ must have $u_1$-value at least $47$.
If its core value were at most $46$, then even full edge coverage and all $n$ vertex goods would give at most
\begin{equation*}
L(46m+3m)+n
=
49Lm+n
<
50Lm+k.
\end{equation*}
Core fact~\ref{core:rigidity} therefore forces $Y_1\cap C$ to have counts $(1,2)$, $Y_2\cap C$ to have counts $(2,3)$, and agent~3 to receive no core good.

Let $Q_i'=Y_i\cap P$ and $q_i'=|Q_i'|$.
Agent~1 must have full edge coverage in $Y$.
If $\cov(Q_1')\le m-1$, then
\begin{equation*}
w_1(Y_1)
\le
L(47m+3m-3)+n
<
50Lm+k,
\end{equation*}
where the last inequality follows from $3L>n+1$ and $k\ge1$.
Hence $\cov(Q_1')=m$, and agent~1's utility comparison implies $q_1'\ge k$.
Agent~3's comparison implies $q_3'\ge n-k$.
Since the vertex goods are partitioned, $q_1'+q_2'+q_3'=n$.
It follows that
\begin{equation*}
q_1'=k,
\quad
q_2'=0,
\quad
q_3'=n-k.
\end{equation*}
All three agents then have exactly the same utility in $Y$ as in $X$, contradicting the requirement that a Pareto domination be strict for at least one agent.
Thus $X$ is \PO.

\proofpart{Reverse direction.}
Assume the constructed instance has an allocation $X=(X_1,X_2,X_3)$ that is both \EFone{} and \PO.

Agent~3 receives no core good.
Otherwise, moving one such good from agent~3 to agent~1 would leave agents~2 and~3 unchanged and would strictly improve agent~1, because agent~1 has positive marginal value for every core good.
This would contradict \PO.
Hence agents~1 and~2 partition the core goods.
Their core split must be Pareto optimal for $u_1,u_2$; otherwise, changing only the core allocation would Pareto improve the full allocation.
By core fact~\ref{core:frontier}, the split is one of the eight displayed rows.

We first eliminate every split listed in~\ref{core:frontier} with
$\sigma_2<0$.
In each such row, $\sigma_2\le-3$.
For every $g\in X_1$, the core part of $X_1\setminus\{g\}$ has $u_2$-value at least three more than agent~2's own core bundle.
For a core good this follows from the definition of $\sigma_2$; for a vertex good the core bundle is unchanged and hence is no smaller than any one-core-good deletion.
The vertex-count terms can reduce the resulting difference by at most $n$.
Therefore
\begin{equation*}
w_2(X_1\setminus\{g\})-w_2(X_2)
\ge
3Lm-n
>
0
\end{equation*}
for every $g\in X_1$, so agent~2 cannot satisfy \EFone{} at any of these rows.

At the three splits $(0,0),(1,0),(1,1)$ listed in~\ref{core:frontier}, agent 1's core EF1 difference is at most $-15$.
For every $g\in X_2$, the coverage term can increase $w_1(X_1)-w_1(X_2\setminus\{g\})$ by at most $3Lm$, and the vertex-count term can increase it by at most $n$.
Consequently,
\begin{equation*}
w_1(X_1)-w_1(X_2\setminus\{g\})
\le
-15Lm+3Lm+n
<
0.
\end{equation*}
Agent~1 therefore cannot satisfy \EFone{} at these rows either.
The only remaining core split is $(1,2)$: agent~1 receives one $A$-good and two $B$-goods, and agent~2 receives the other two $A$-goods and three $B$-goods.

For $i\in\{1,2,3\}$, define $Q_i=X_i\cap P$, $q_i=|Q_i|$, and $c_i=\cov(Q_i)$.
Agent~1's own utility is $w_1(X_1)=L(47m+3c_1)+q_1$.
For every $g\in X_2$,
\begin{equation*}
w_1(X_2\setminus\{g\})
\ge
L(50m+3c_2)+q_2-1.
\end{equation*}
Indeed, if $g$ is a core good, the remaining core bundle has $u_1$-value $50$, while coverage and the vertex count are unchanged.
If $g$ is a vertex good, the core value remains $53$, coverage falls by at most $m$, and the vertex count falls by one; hence $53m+3(c_2-m)=50m+3c_2$.

Agent~1's \EFone{} condition provides some $g\in X_2$ such that $w_1(X_1)\ge w_1(X_2\setminus\{g\})$.
Combining the above facts gives us
\begin{equation*}
3L(c_1-c_2-m)\ge q_2-q_1-1.
\end{equation*}
Since $c_1\le m$ and $c_2\ge0$, unless $c_1=m$ and $c_2=0$, the integer $c_1-c_2-m$ is at most $-1$.
The left-hand side would then be at most $-3L$, whereas the right-hand side is at least $-n-1$.
This contradicts $3L>n+1$.
Therefore $c_1=m$ and $c_2 = 0$.
The vertices represented by $Q_1$ consequently form a vertex cover of $G$.

Finally, agent~3's \EFone{} comparison with agent~1 gives $q_3\ge q_1-1$.
The set $Q_1$ is nonempty because it covers a graph with at least one edge.
Deleting a vertex good from $X_1$ leaves agent~3 value $L(q_1-1)$, while deleting a core good leaves value $Lq_1$; hence the best deletion for agent~3 is a vertex good.
Using $q_1+q_2+q_3=n$, we obtain $2q_1+q_2\le n+1$.
Since $n$ is even and $q_2\ge0$, this implies $q_1\le n/2$.
Thus $G$ has a vertex cover of size at most $n/2$, completing the reduction.
\end{proof}

Note that Theorem~\ref{thm:hardness} does not establish NP-hardness when all three valuations are strictly increasing, since agent~3 assigns zero marginal value to the eight core goods. Rather, it eliminates all input-dependent zero marginals: every vertex good has positive marginal value for every agent, and the only zero marginals are the eight fixed pairs $(3,g)$ with $g \in C$. In contrast, in the reduction of Chandramouleeswaran and Nimbhorkar~\cite[Section~5]{CN26}, one agent assigns zero marginal value to every vertex good, so the number of such pairs grows with the graph.

\section{Conclusion}

We determine the exact two-agent threshold for incompatibility between \EFone{} and \PO{} under strictly increasing valuations. Every instance with at most seven goods admits an allocation satisfying both properties, without any submodularity assumption. In contrast, with eight goods, we construct normalized, integer-valued, strictly increasing, submodular valuations for which every \EFone{} allocation is strictly Pareto dominated. Thus, strictly positive marginal values do not restore compatibility in general, but they rule out every smaller two-agent counterexample.

We also strengthen the three-agent NP-hardness frontier for monotone submodular valuations. In our reduction, every input-dependent vertex good has strictly positive marginal value for every agent, while all zero marginals are confined to one agent on the eight fixed core goods. This leaves two natural directions for future work: whether NP-hardness persists when all three valuations are strictly increasing, and, for three or more agents, the minimum number of goods required for \EFone{} and \PO{} to be incompatible under strictly increasing valuations.

\bibliographystyle{plainnat}
\bibliography{bib}

\begin{thebibliography}{18}
\providecommand{\natexlab}[1]{#1}
\providecommand{\url}[1]{\texttt{#1}}
\expandafter\ifx\csname urlstyle\endcsname\relax
  \providecommand{\doi}[1]{doi: #1}\else
  \providecommand{\doi}{doi: \begingroup \urlstyle{rm}\Url}\fi

\bibitem[Amanatidis et~al.(2023)Amanatidis, Aziz, Birmpas, Filos-Ratsikas, Li, Moulin, Voudouris, and Wu]{AmanatidisABFLMVW23}
Georgios Amanatidis, Haris Aziz, Georgios Birmpas, Aris Filos-Ratsikas, Bo~Li, Herv\'{e} Moulin, Alexandros~A. Voudouris, and Xiaowei Wu.
\newblock Fair division of indivisible goods: Recent progress and open questions.
\newblock \emph{Artificial Intelligence}, 322\penalty0 (C), 2023.

\bibitem[Aziz et~al.(2022)Aziz, Caragiannis, Igarashi, and Walsh]{aziz2022goodschores}
Haris Aziz, Ioannis Caragiannis, Ayumi Igarashi, and Toby Walsh.
\newblock Fair allocation of indivisible goods and chores.
\newblock \emph{Autonomous Agents and Multi-Agent Systems}, 36\penalty0 (3), 2022.

\bibitem[Barman and Suzuki(2026)]{BarmanSuzuki2026}
Siddharth Barman and Mashbat Suzuki.
\newblock Compatibility of fairness and {Nash} welfare under subadditive valuations.
\newblock In \emph{Proceedings of the 2026 Annual ACM-SIAM Symposium on Discrete Algorithms (SODA)}, pages 1724--1746, 2026.

\bibitem[Barman et~al.(2018)Barman, Krishnamurthy, and Vaish]{BarmanKV18}
Siddharth Barman, Sanath~Kumar Krishnamurthy, and Rohit Vaish.
\newblock Finding fair and efficient allocations.
\newblock In \emph{Proceedings of the 19th ACM Conference on Economics and Computation (EC)}, pages 557--574, 2018.

\bibitem[Benabbou et~al.(2021)Benabbou, Chakraborty, Igarashi, and Zick]{BCIZ21}
Nawal Benabbou, Mithun Chakraborty, Ayumi Igarashi, and Yair Zick.
\newblock Finding fair and efficient allocations for matroid rank valuations.
\newblock \emph{ACM Transactions on Economics and Computation}, 9:\penalty0 1--41, 2021.

\bibitem[Bhaskar et~al.(2021)Bhaskar, Sricharan, and Vaish]{BhaskarEtAl2021}
Umang Bhaskar, A.~R. Sricharan, and Rohit Vaish.
\newblock On approximate envy-freeness for indivisible chores and mixed resources.
\newblock In \emph{Approximation, Randomization, and Combinatorial Optimization. Algorithms and Techniques (APPROX/RANDOM)}, pages 1:1--1:23, 2021.

\bibitem[Brandl et~al.(2026)Brandl, Suksompong, and Teh]{Brandl2026}
Florian Brandl, Warut Suksompong, and Nicholas Teh.
\newblock Fair division with binary valuations: Characterizations.
\newblock In \emph{Proceedings of the 19th International Symposium on Algorithmic Game Theory (SAGT)}, 2026.
\newblock Extended version available at arXiv:2607.10064.

\bibitem[Budish(2011)]{Budish2011}
Eric Budish.
\newblock The combinatorial assignment problem: Approximate competitive equilibrium from equal incomes.
\newblock \emph{Journal of Political Economy}, 119\penalty0 (6):\penalty0 1061--1103, 2011.

\bibitem[Caragiannis et~al.(2019)Caragiannis, Kurokawa, Moulin, Procaccia, Shah, and Wang]{CaragiannisKMPPSW19}
Ioannis Caragiannis, David Kurokawa, Hervé Moulin, Ariel~D. Procaccia, Nisarg Shah, and Junxing Wang.
\newblock The unreasonable fairness of maximum {N}ash welfare.
\newblock \emph{ACM Transactions on Economics and Computation}, 7\penalty0 (3):\penalty0 12:1--12:32, 2019.

\bibitem[Chandramouleeswaran and Nimbhorkar(2026)]{CN26}
Harish Chandramouleeswaran and Prajakta Nimbhorkar.
\newblock Nonexistence of simultaneously {EF1} and pareto optimal allocations for submodular valuations.
\newblock \emph{arXiv preprint arXiv:2607.18220}, 2026.

\bibitem[Conitzer and Sandholm(2006)]{CS06}
Vincent Conitzer and Tuomas Sandholm.
\newblock Computing the optimal strategy to commit to.
\newblock In \emph{Proceedings of the 7th ACM Conference on Electronic Commerce (EC)}, pages 82--90, 2006.

\bibitem[Hosseini et~al.(2025)Hosseini, Narang, and W{\k{a}}s]{HosseiniEtAl2025}
Hadi Hosseini, Shivika Narang, and Tomasz W{\k{a}}s.
\newblock Fair distribution of delivery orders.
\newblock \emph{Artificial Intelligence}, 347:\penalty0 104389, 2025.

\bibitem[Lim et~al.(2026)Lim, Neoh, and Teh]{Lim2026EFXcost}
Eugene Lim, Tzeh~Yuan Neoh, and Nicholas Teh.
\newblock The cost of {EFX}: Generalized-mean welfare and complexity dichotomies with few surplus items.
\newblock \emph{arXiv preprint arXiv:2601.12849}, 2026.

\bibitem[Lipton et~al.(2004)Lipton, Markakis, Mossel, and Saberi]{LiptonMMS04}
Richard~J. Lipton, Evangelos Markakis, Elchanan Mossel, and Amin Saberi.
\newblock On approximately fair allocations of indivisible goods.
\newblock In \emph{Proceedings of the 5th ACM Conference on Electronic Commerce (EC)}, pages 125--131, 2004.

\bibitem[Mackenzie and Suzuki(2026)]{MS26}
Simon Mackenzie and Mashbat Suzuki.
\newblock When one good is not enough: {EF1} and pareto optimality are not compatible for submodular valuations.
\newblock \emph{arXiv preprint arXiv:2607.17811}, 2026.

\bibitem[Mahara(2026)]{Mahara2026}
Ryoga Mahara.
\newblock Existence of fair and efficient allocation of indivisible chores.
\newblock In \emph{Proceedings of the 37th ACM-SIAM Symposium on Discrete Algorithms (SODA)}, pages 6742--6766, 2026.

\bibitem[Plaut and Roughgarden(2020)]{plaut2020almost}
Benjamin Plaut and Tim Roughgarden.
\newblock Almost envy-freeness with general valuations.
\newblock \emph{SIAM Journal on Discrete Mathematics}, 34\penalty0 (2):\penalty0 1039--1068, 2020.

\bibitem[Teh(2026)]{Teh2026}
Nicholas Teh.
\newblock Computing fair and efficient indivisible chore allocations with bounded surplus.
\newblock In \emph{Proceedings of the 19th International Symposium on Algorithmic Game Theory (SAGT)}, 2026.

\end{thebibliography}

\appendix

\section{Marginal Arrays for the Eight-Good Instance}\label{app:marginals}

For rows indexed by $x$ and columns indexed by $y$, the $A$-marginal arrays are
\begin{equation}
\begin{array}{c|rrrrrr}
\Delta_A u_1&0&1&2&3&4&5\\ \hline
x=0&25&25&21&11&7&1\\
x=1&3&3&3&3&1&1\\
x=2&1&1&1&1&1&1
\end{array}
\quad
\begin{array}{c|rrrrrr}
\Delta_A u_2&0&1&2&3&4&5\\ \hline
x=0&17&15&13&7&3&3\\
x=1&15&11&7&7&1&1\\
x=2&15&11&7&1&1&1
\end{array}.
\label{eq:a-marginals}
\end{equation}
The $B$-marginal arrays are
\begin{equation}
\begin{array}{c|rrrrr}
\Delta_B u_1&0&1&2&3&4\\ \hline
x=0&13&13&13&7&7\\
x=1&13&9&3&3&1\\
x=2&13&9&3&1&1\\
x=3&13&9&3&1&1
\end{array}
\quad
\begin{array}{c|rrrrr}
\Delta_B u_2&0&1&2&3&4\\ \hline
x=0&13&13&13&11&1\\
x=1&11&11&7&7&1\\
x=2&7&7&7&1&1\\
x=3&3&3&1&1&1
\end{array}.
\label{eq:b-marginals}
\end{equation}
Every entry is positive.
Within each array, every row and every column is nonincreasing, verifying the hypotheses of Lemma~\ref{lem:count-criterion}.

\section{\EFone{} Comparison Differences}\label{app:ef1-differences}

Recall that $d_1(x,y)$ is the minimum value, according to agent~1, of agent~2's bundle after one feasible deletion, and $d_2(x,y)$ is defined symmetrically for agent~2.
For nonautomatic comparisons,
\begin{equation*}
\sigma_1(x,y)=u_1(x,y)-d_1(x,y),
\quad
\sigma_2(x,y)=u_2(3-x,5-y)-d_2(x,y).
\end{equation*}
Denote $\star$ when the comparison is automatic because the envied bundle is empty.
Then
\begin{equation}
\begin{array}{c|rrrrrr}
\sigma_1&y=0&y=1&y=2&y=3&y=4&y=5\\ \hline
x=0&-55&-41&-25&-3&17&25\\
x=1&-29&-15&-3&9&25&29\\
x=2&-25&-5&11&27&41&55\\
x=3&-17&3&25&41&55&\star
\end{array}
\label{eq:sigma1-table}
\end{equation}
and
\begin{equation}
\begin{array}{c|rrrrrr}
\sigma_2&y=0&y=1&y=2&y=3&y=4&y=5\\ \hline
x=0&\star&55&41&27&11&-3\\
x=1&55&41&27&7&-7&-19\\
x=2&37&25&7&-7&-25&-37\\
x=3&19&11&-7&-27&-41&-55
\end{array}.
\label{eq:sigma2-table}
\end{equation}
The pairs $(x,y)$ for which both relevant differences are nonnegative are exactly the four splits in \eqref{eq:ef1-splits}.

\end{document}